 \documentclass[preprints,article,accept,moreauthors,pdftex,10pt,a4paper]{Definitions/mdpi} 

\firstpage{1} 
\pubvolume{xx}
\issuenum{1}
\articlenumber{5}
\pubyear{2018}
\copyrightyear{2018}
\newtheorem{lemma}[theorem]{Lemma}
\history{December 12, 2018} 

\pdfoutput=1

\usepackage{array}

\Title{Non-classical Correlations in $n$-Cycle Setting}

\Author{Kishor Bharti $^{1}$, Maharshi Ray $^{1}$ and Leong Chuan Kwek $^{1,2,3,}$*} 

\AuthorNames{Kishor Bharti, Maharshi Ray and Leong Chuan Kwek}

\address{%
$^{1}$ \quad Centre for Quantum Technologies, National University of Singapore\\
$^{2}$ \quad MajuLab, CNRS-UNS-NUS-NTU International Joint Research Unit, Singapore UMI 3654, Singapore\\
$^{3}$ \quad National Institute of Education, Nanyang Technological University, Singapore 637616, Singapore}

\corres{Correspondence: cqtklc@gmail.com}

\abstract{ We explore non-classical correlations in $n$-cycle setting. In particular, we focus on correlations manifested by Kochen-Specker-Klyachko box (KS box), scenarios involving $n$-cycle non-contextuality inequalities and Popescu-Rohlrich  boxes (PR box). We provide the criteria for optimal classical simulation of a KS box of arbitrary $n$ dimension. The non-contextuality inequalities are analysed for $n$-cycle setting, and the condition for the quantum violation for odd as well as even $n$-cycle is discussed. We offer a simple extension of even cycle non-contextuality inequalities to the continuous variable case. Furthermore, we simulate a generalized PR box using KS box and provide some interesting insights. Towards the end, we discuss a few possible interesting open problems for future research.}


\begin{document}


\section{Introduction}
 The quantum mechanical description of nature is incompatible with any local hidden variable theory and hence consequently is said to exhibit Bell nonlocality \cite{Bell64}. This counter-intuitive phenomenon finds applications in various quantum information  processing tasks like randomness certification \cite{pironio2010random}, self-testing \cite{popescu1992generic,Yao_self, summers1987bell, tsirel1987quantum} and distributed computing \cite{cleve1997substituting}. The Bell nonlocality can be thought of as a particular case of another not-so-famous phenomenon, referred to as contextuality \cite{CSW, amaral2018graph, kochen1975problem}. Recently, contextuality has been shown to be useful for quantum cryptography \cite{JKA, Cabello_QKD} and self-testing \cite{ST2018}. Furthermore, it has been shown to be a crucial resource responsible for various models of quantum computing including measurement based quantum computing \cite{raussendorf2013contextuality} and fault-tolerant quantum computing \cite{Howard2014}. An in-depth study of non-locality and contextuality is required to harness these features for applications as well as to deepen our understanding of nature. In this regard, we study various nonlocal and contextual resources in various sections of this paper, which we describe subsequently.


In section \ref{KS}, we study the less known Kochen Specker Klyachko box (or KS box) for the $n$-dimensional case. The box as mentioned earlier was first introduced by Jeffrey Bub {\it{et. al.}} in 2009 \cite{bub2009contextuality} and was analysed for 5-dimensional case. We study the box for general $n$-dimensional case and provide the optimal classical strategy as well as corresponding success probability for simulating the box using classical resources.


The Bell nonlocal nature of theories can be witnessed via the violation of certain inequalities, referred to as Bell inequalities and non-contextuality inequalities in the general case of contextuality \cite{amaral2018graph}. In section \ref{Ineq}, we study $n$-cycle contextuality scenario and corresponding non-contextuality inequalities. We explore the odd cycle generalisation of the well-known Klyachko-Can-Binicio\u{g}lu-Shumovsky (KCBS) inequality \cite{KCBS, Liang} and even cycle generalisation of Clauser-Horne-Shimony-Holt
(CHSH) inequality \cite{CHSH}. Following the construction provided by Ara{\'{u}}jo {\it{et. al.}} \cite{Araujo2013}, we discuss the necessary and sufficient condition for the violation of the generalised KCBS inequality and necessary condition for the violation of even-cycle generalisation of CHSH inequality. Note that the even cycle generalisation of CHSH inequality are similar to Braunstein-Caves inequalities \cite{braunstein1990wringing} which have been heavily investigated in the literature. Following the work of Arora {\it{et. al.}} \cite{arora2015proposal}, we provide a simple extension of these even cycle non-contextuality inequalities to the continuous variable case.

Within no-signalling theories, the maximum violation of CHSH inequality is obtained by Popescu-Rohlrich box, also known as PR box \cite{popescu1994quantum}. The PR-box and its analogue for even-cycle generalisation of CHSH inequality are the contents of section \ref{PR}. In their seminal work \cite{bub2009contextuality}, Bub {\it{et. al.}} studied the simulation of a PR box using KS box. We extend the idea to arbitrary dimensional KS box and PR box.  We study the joint probability distribution for the KS box and find the criteria for the violation of even-cycle generalisation of CHSH inequality. Finally, we conclude in section \ref{conclusion}.

The paper connects generalized PR boxes, arbitrary dimensional KS boxes and $n$-cycle noncontextuality inequalities and thus provides the pathway for the study of these contextual and nonlocal resources at their junction.

\section{Simulating KS Box} \label{KS}
\begin{Definition}
An $N$-dimensional Kochen Specker Klyachko box or KS box, defined in \cite{bub2009contextuality} is a no-signalling resource with two inputs, $x,y\in\left\{ 1,2, \cdots, N\right\} $ and two outputs $a,b\in\{0,1\},$  which satisfies the following constraints:
\begin{enumerate}
\item $a=b$ if $x=y$, and 
\item $a.b=0$ if $x\neq y.$
\end{enumerate}
\end{Definition}
A KS-box with marginal probability $p$ for the output '$1$` is referred to as $KS_{p}$ box. For example, the fraction of `$1$'s in a $KS_{\frac{1}{5}}$ box is $\frac{1}{5}$. We shall refer to the KS-box condition corresponding to $a.b = 0$ for unequal   inputs as $\bot.$ It is not possible to simulate the KS box statistics with full accuracy for arbitrary $p$ \cite{bub2009contextuality}. We want to find the probability of successful simulation of KS box statistics for various strategies. Consider an $N$-gon with a $0/1$ assignment to its vertices. A $0/1$ assignment with $M$ `$1$'s for a given $N$-gon corresponding to an $N$-dimensional KS box is referred to as a chart of degree $M$, in short $C_M.$  For example chart $C_1$ for a $5$ dimensional KS box will assign `$1$' to one of the vertices and '$0$` to the rest. The spatially separated parties say, Alice and Bob, will share such charts and using shared randomness decide which chart to use. Clearly using chart $C_0$ and $C_1$ will always satisfy the $\bot$ condition. All other charts will violate the $\bot$ condition up to varying proportion.
\begin{figure}[H]
\label{Chart_2}
\centering
\includegraphics[width=0.25\textwidth]{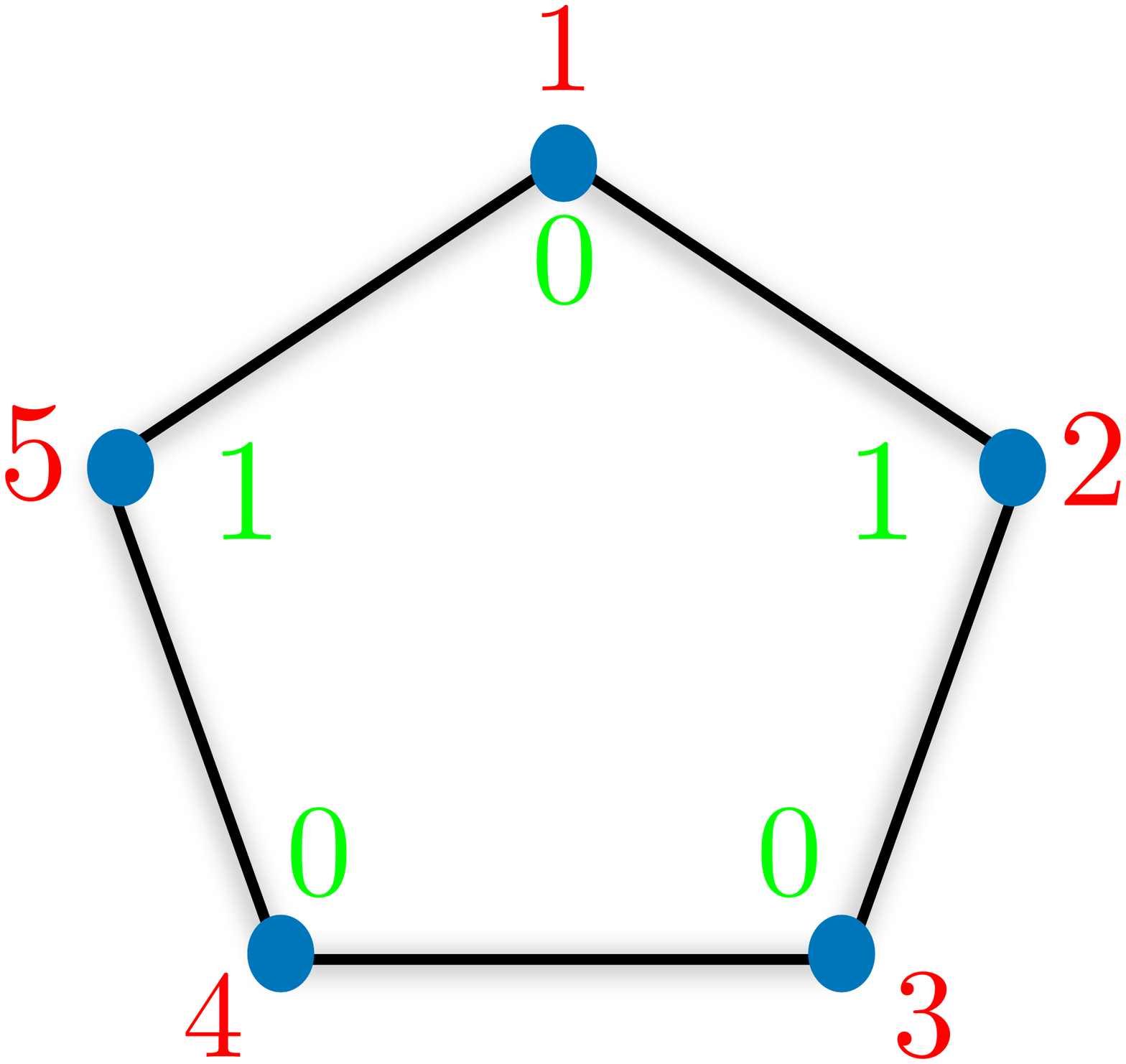}
\centering
\vspace{-1cm}
\caption{Chart $C_2$ for a five-dimensional KS box corresponds to two `$1$'s and three `$0$'s. The red entries correspond to inputs and the outputs are in green. The above chart fails to simulate the KS box statistics when the inputs are $2$ and $5$.}

\end{figure}
Simulating a $KS_p$-box essentially requires the satisfaction of the $\bot$ conditions along with the marginal condition. The use of charts already guarantees equal outputs for same inputs. 

\vspace{0.5cm}
\begin{lemma}
Given the chart $C_M$, the probability of successful simulation of the $\bot$ condition is given by 
\begin{equation}
P_{\bot}\left(C_{M}\right)=\frac{N^{2}-M^{2}+M}{N^{2}}.\label{eq:Succ_SIm}
\end{equation}
\end{lemma}
\begin{proof}[Proof]
For an $N$-dimensional KS box, the total number of possible input pairs for Alice and Bob are $N^{2}.$ If they use the chart $C_{M}$ to simulate the KS box, then the probability of failure corresponds to the probability of choosing different inputs with output $1$. The number of such edges (with ordering) whose vertices correspond to output $1$ is $M(M-1).$ 
Thus, the probability of successful simulation is
\begin{equation}
1 - \frac{M(M-1)}{N^2} = \frac{N^{2}-M^{2}+M}{N^{2}}.
\end{equation}
This completes the proof.
\end{proof}

%
%
%
For $p \leq \frac{1}{N},$ Alice and Bob can use chart $C_0$ and $C_1$ to simulate the KS-box. However, we observe that in order to satisfy the marginal constraints for $p > \frac{1}{N},$ one needs to use charts of higher degree, which in turn violates the $\bot$ conditions. Therefore perfect classical simulation of the $KS_p$-box only exists for $p \leq \frac{1}{N}.$ We now fix a $p \leq 0.5$ and compute the optimal classical simulation probability of the $KS_p$-box. 
Now we present our result concerning the optimal probability of successful simulation  for an $N$-dimensional $KS_p$ box for arbitrary $p$.

%
%
%
\begin{Theorem}
For a given $p \leq 0.5,$ the charts $C_{M-1}$ and $C_{M}$ (only chart $C_M$ in case $Np$ is an integer) are optimal for simulating $N$-dimensional $KS_p$-box, where $M = [Np]$ $\left(\text{integral part}\right)$ and the optimal probability of simulation is given by  $$P_{optimal}\left(M,N,p\right)=1- \frac{(2Np-M)(M-1)}{N^2}.$$
\end{Theorem}
\begin{proof}[Proof]
Assume that Alice and Bob play the charts $C_i$ with probability $p_i,$ for $i \in \mathbb{Z}^{\geq},$ i.e. the set of non-negative integers. For a given probability distribution $\{p_i\}$ over charts, the probability of successful simulation of $KS_p$-box  is given by $\sum_i p_iP_{\bot}\left(C_{i}\right).$ Hence the optimal simulation probability is given by the following linear program :
 \begin{align*}
&\max_{\{p_i\}} \sum_i p_iP_{\bot}\left(C_{i}\right) &(\text{success probability})\\
& s.t \sum_i p_ii = Np &(\text{mean condition}) \\
& \sum_i p_i = 1, p_i \geq 0 \, \forall i &(\text{valid probability})
\end{align*}
Now observe that the objective function is 
\begin{align*}
\sum_i p_iP_{\bot}\left(C_{i}\right) &= \frac{1}{N^{2}} \sum_i p_i \left(N^{2}-i^{2}+i\right)\\
&= 1+\frac{p}{N} -\frac{1}{N^2} \sum_i p_i i^2, 
\end{align*}
where in the second equality we used the mean condition along with the valid probability condition. Hence, maximising the objective function corresponds to minimising the variance term with respect to the probability distribution $\{p_i\}.$ The optimisation problem of minimising the variance of a random variable defined on a set of non-negative integral points, over all possible probability distributions, for a fixed given mean, has support size at most two. This can be seen easily using the Karush–Kuhn–Tucker conditions. Specifically, if the mean ($Np$) is an integer (say $=M$), the least variance solution will be $p_M =1$ and $p_i = 0, \forall i \neq M.$ For the case when the mean is not an integer, the least variance solution corresponds to a support containing $M-1$ and $M$, with $M = [Np]$, which follows from simple convexity arguments. With this support, we can compute $p_{M-1}$ and $p_M$ using the mean condition, which evaluates to $p_{M-1} = M-Np$ and $p_M = Np- M +1.$ Plugging this into the success probability function gives us the optimal simulation probability of the $KS_p$-box 
$$P_{optimal}\left(M,N,p\right)=\frac{1}{N^{2}}\left(2Np-2NpM+N^{2}+M^{2}-M\right) = 1- \frac{(2Np-M)(M-1)}{N^2}$$ 
This completes the proof.\end{proof}
\begin{Remark}
Simulation efficiency decreases monotonically with dimension. Five-dimensional KS box is optimal for efficient simulation of any arbitrary marginal $p$.
\end{Remark}

\begin{Remark}
For large $N$, the simulation efficiency tends to $1-p^2.$
\end{Remark}

Having studied the KS box, we move next to the $n$-cycle non-contextuality inequalities.
%
%
%
%
%
%

\section{Analysing n-Cycle Non-Contextuality Inequalities}\label{Ineq}
We analyse the n-cycle generalisation of KCBS and CHSH inequalities. The KCBS inequality is a state-dependent non-contextuality inequality with five dichotomic measurements with $0/1$ outcome. The combination of measurement and corresponding outcome constitutes an event. For example, $\left(a\vert i \right)$ is an event which corresponds to getting outcome ``a'' for measurement ``i''.  Let us represent the probability of getting outcome ``$1$'' given the input was ``$i$'' as $P\left(1\vert i \right)$.  The events follow exclusivity relation according to a graph, referred to as exclusivity graph (a pentagon in this case). The exclusivity relation induces following constraint:
%
\begin{equation}
\label{exclusivity}
P\left(1\vert i \right) + P\left(1\vert j \right) \le 1,
\end{equation}
$\forall i,j \in E$. The KCBS inequality corresponds to sum of  probabilities assigned to five events of the kind $\left(1\vert i \right)$ with exclusivity relation following a pentagon. 
\begin{figure}[H]
\label{Chart_2}
\centering
\includegraphics[width=0.35\textwidth]{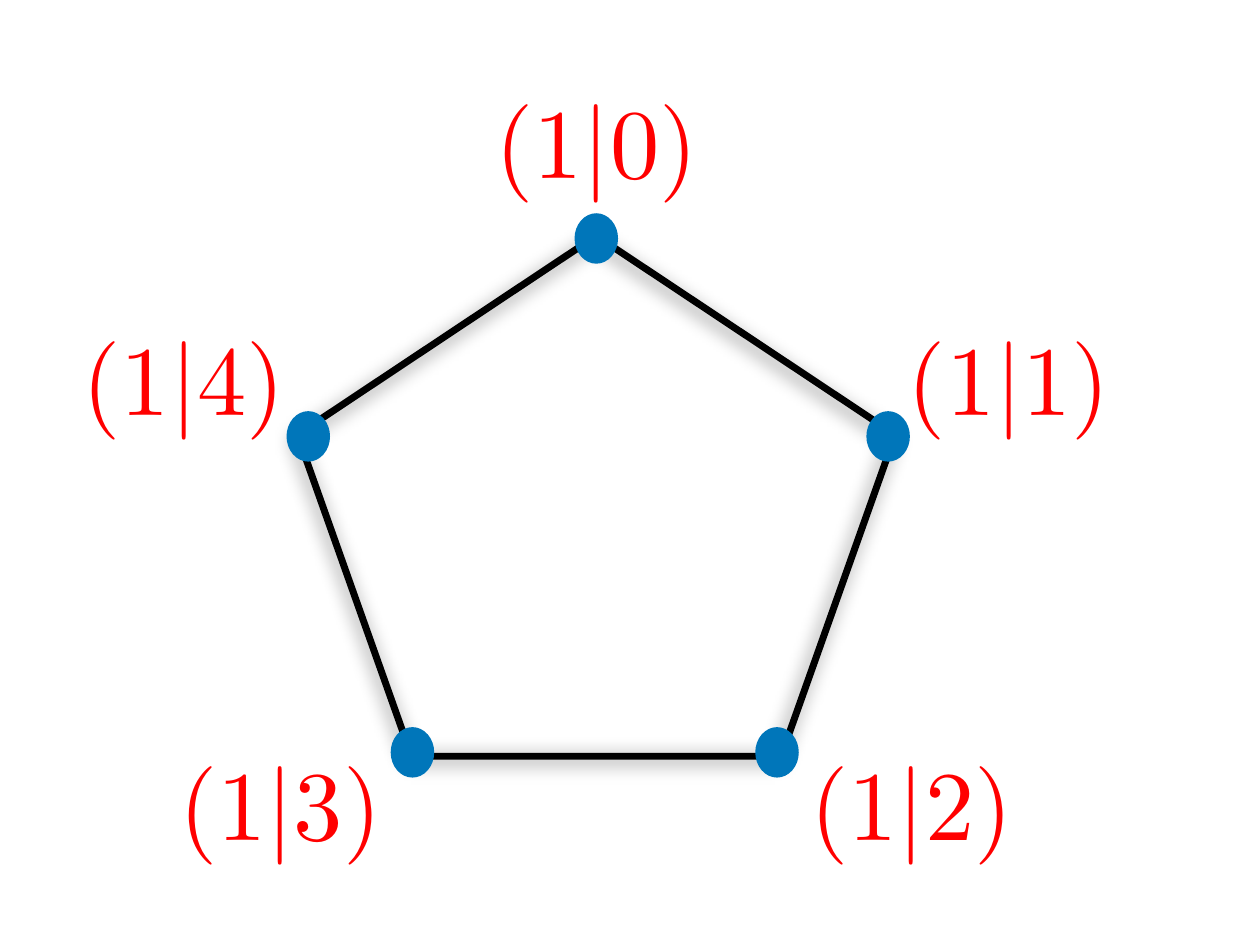}
\centering
\vspace{0 cm}
\caption{The exclusivity graph corresponding to KCBS inequality is a pentagon. The inequality involves five events of type $\left(1\vert i\right)$ where $i \in \{0,1,2,3,4\}. $ The bound on the inequality for non-contextual hidden variable theories is $2.$ Quantum theory achieves up to $\sqrt{5}$ and thus manifests the contextual nature of quantum theory.}

\end{figure}
For a non-contextual hidden variable theory, the bound on the inequality is 2. Formally, the KCBS inequality is given by 
\begin{equation}
\label{eq:KCBS_5}
\sum_{i=0}^{4}P\left( 1\vert i\right)\leq 2.
\end{equation}
The inequality in \eqref{eq:KCBS_5} has been further extended to general odd cycle, which is
\begin{equation}
\label{eq:KCBS_n}
\sum_{i=0}^{n-1}P\left( 1\vert i\right)\leq \frac{n-1}{2}.
\end{equation}
The odd cycle generalisation of KCBS inequality has been studied extensively in literature \cite{Liang,bharti2018simple,Araujo2013,ST2018}. Surprisingly, $\frac{n-1}{2}$ corresponds to independence number of the graph for odd cycle graph \cite{knuth1994sandwich,Liang, CSW}. The maximum quantum violation for generalised KCBS inequality corresponds to Lov{\'a}sz theta number (denoted by $\vartheta\left(G\right)$), which is $\frac{n\cos\left(\frac{\pi}{n}\right)}{1+\cos\left(\frac{\pi}{n}\right)}.$
%
We will represent the density matrices in the standard basis $\{\vert i\rangle\}$ with matrix elements given by $\rho_{ij}=\langle i\vert\rho\vert j\rangle.$ For the odd n-cycle generalisation of KCBS inequality, the projectors corresponding to the optimal quantum violation are given by 

\[
\Pi_{j}=\vert\psi_{j}\rangle\langle\psi_{j}\vert
\]
where
\[
\vert\psi_{j}\rangle=\left(\sin\left(\theta\right)\cos\left(\frac{j\pi\left(n-1\right)}{n}\right),\sin\left(\theta\right)\sin\left(\frac{j\pi\left(n-1\right)}{n}\right),\cos\left(\theta\right)\right)^{T}
\]
and $\cos^{2}\left(\theta\right)=\frac{\cos\left(\frac{\pi}{n}\right)}{1+\cos\left(\frac{\pi}{n}\right)}.$ Now we present the condition under which a qutrit will violate the generalised KCBS inequality for the above measurement settings.  
\begin{Proposition}
A qutrit violates the odd n cycle generalization of KCBS noncontextuality inequality if and only if
$\rho_{33}\ge\left(\frac{\cos\left(\frac{\pi}{n}\right)\left(n-1\right)-1}{n\left(2\cos\left(\frac{\pi}{n}\right)-1\right)}\right)$.
\end{Proposition}
\begin{proof}[Proof]

The generalised KCBS operator for the odd n-cycle scenario can be defined as
\[
K_{n}=\sum_{j=1}^{n}\Pi_{j}.
\]
Adding all the projectors ($\Pi_{j}$s), we get

\[
K_{n}=\sum_{i=1}^{3}k_{i}\vert\phi_{i}\rangle\langle\phi_{i}\vert
\]
where 
\[
\vert\phi_{1}\rangle=\begin{pmatrix}1\\
0\\
0
\end{pmatrix},\vert\phi_{2}\rangle=\begin{pmatrix}0\\
1\\
0
\end{pmatrix},\vert\phi_{3}\rangle=\begin{pmatrix}0\\
0\\
1
\end{pmatrix}
\]
and 
\[
k_{1}=\frac{1}{1+\cos\left(\frac{\pi}{n}\right)}\sum_{j=1}^{n}\cos^{2}\left(\frac{j\pi\left(n-1\right)}{n}\right)
\]

\[
k_{2}=\frac{1}{1+\cos\left(\frac{\pi}{n}\right)}\sum_{j=1}^{n}\sin^{2}\left(\frac{j\pi\left(n-1\right)}{n}\right)
\]

\[
k_{3}=n\cos^{2}\left(\theta\right)=\frac{n\cos\left(\frac{\pi}{n}\right)}{1+\cos\left(\frac{\pi}{n}\right)}.
\]
Since $\sum_{j}\cos^{2}\left(\frac{j\pi\left(n-1\right)}{n}\right)=\sum_{j}\sin^{2}\left(\frac{j\pi\left(n-1\right)}{n}\right)=\frac{n}{2},$
we get

\[
k_{1}=k_{2}=\frac{n}{2\left(1+\cos\left(\frac{\pi}{n}\right)\right)}
\]

\[
k_{3}=\frac{n\cos\left(\frac{\pi}{n}\right)}{1+\cos\left(\frac{\pi}{n}\right)}
\]
The odd n-cycle noncontextuality inequality is written as

\[
\langle K_n\rangle\le\frac{n-1}{2},
\]
where $\langle K_n\rangle$ corresponds to the expectation value of the generalised KCBS operator with respect to the underlying preparation.
In terms of quantum expectation, the inequality is given by 
\[
\text{Tr}\left(K_n\rho\right)\le\frac{n-1}{2}.
\]
Note that the generalised KCBS operator is diagonal in standard basis and leads to the following simplification:
\[
\frac{n}{2\left(1+\cos\left(\frac{\pi}{n}\right)\right)}\left[\rho_{11}+\rho_{22}\right]+\frac{n\cos\left(\frac{\pi}{n}\right)}{1+\cos\left(\frac{\pi}{n}\right)}\left[\rho_{33}\right]\le\frac{n-1}{2}.
\]
Since the trace of a density matrix is always 1, the condition for the violation
of odd n-cycle non-contextuality inequality becomes;

\[
\rho_{33}>\frac{\left(n-1\right)\left(1+\cos\left(\frac{\pi}{n}\right)\right)-n}{2\cos\left(\frac{\pi}{n}\right)-1}.
\]
Simplifying the above expression, we get
\begin{equation}
\rho_{33}>\frac{\cos\left(\frac{\pi}{n}\right)\left(n-1\right)-1}{n\left(2\cos\left(\frac{\pi}{n}\right)-1\right)}.
\end{equation}
This completes the proof.
\end{proof}
\begin{Remark}
We can see that the set of quantum states for qutrits, which can violate
odd n cycle noncontextuality inequality, shrinks as we increase $n.$
In the infinite $n$ scenario, the only qutrit which violates the
inequality is the pure state $\vert\psi\rangle=\left(0,0,1\right)^{T}$!

\begin{figure}[H]
\begin{center}\includegraphics[scale=0.23]{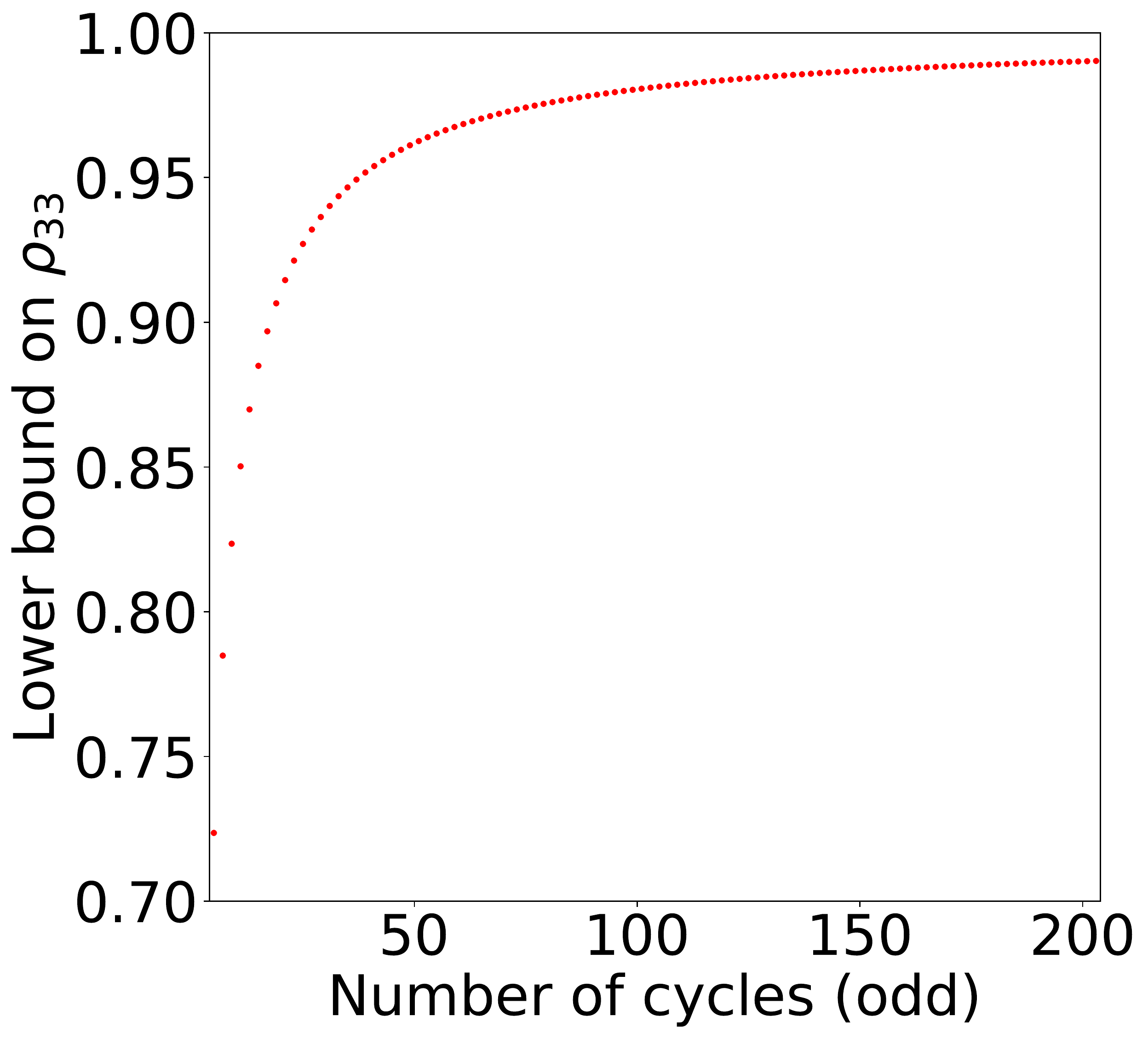}\end{center}\caption{The condition for the quantum violation of the odd $n$ cycle generalisation of KCBS inequality is computed. Lower bound on $\rho_{33}$ for odd n cycle graph has been plotted as a function of $n$. The set of states which can violate the KCBS inequality corresponding to optimal measurement setting shrinks as we increase $n$.}
\end{figure}
\end{Remark}
%
The $n$-cycle generalization of CHSH inequality is referred to as chained Bell inequality \cite{Araujo2013, braunstein1990wringing}. The even $n$-cycle scenario has $n$ measurements i.e $\{X_1,X_2,\cdots,X_n\}$. Each of these are dichotomic measurements with possible outcomes $\pm1.$ 
The chained Bell inequality of cycle $n$ is given by
\begin{equation}
\sum_{j=1}^{n-1}\left\langle X_{j}X_{j+1}\right\rangle -\left\langle X_{n}X_{1}\right\rangle \le n-2. \label{chained-CHSH}
\end{equation}

The optimal construction \cite{Araujo2013} for violation of this inequality corresponds to $X_{j}=\widetilde{X_{j}}\otimes\mathbb{I}$ for even $j$ and $X_{j}=\mathbb{I}\otimes\widetilde{X_{j}}$ for
odd $j$, where 
\begin{equation}
\widetilde{X_{j}}=\cos\left(\frac{j\pi}{n}\right)\sigma_{x}+\sin\left(\frac{j\pi}{n}\right)\sigma_{z}. \label{Ops}
\end{equation}

We now provide the necessary condition for the quantum violation of a chained Bell inequality corresponding to optimal quantum measurement settings.
\begin{Proposition}
For a given two qubit state, the necessary condition for the quantum violation
of chained Bell inequality of cycle $n$ is given by the difference of
its extremal eigenvalues i.e.
\end{Proposition}
\begin{equation}
\lambda_{1}-\lambda_{4} > \frac{n-2}{n}.
\end{equation}
\begin{proof}[Proof]
For even $j$,

\[
X_{j}X_{j+1}=\widetilde{X_{j}}\otimes\widetilde{X_{j+1}}
\]

\begin{equation}
=\left[\cos\left(\frac{j\pi}{n}\right)\sigma_{x}+\sin\left(\frac{j\pi}{n}\right)\sigma_{z}\right]\otimes\left[\cos\left(\frac{\left(j+1\right)\pi}{n}\right)\sigma_{x}+\sin\left(\frac{\left(j+1\right)\pi}{n}\right)\sigma_{z}\right].\label{even-1}
\end{equation}
Similarly for odd $j$,

\begin{equation}
X_{j}X_{j+1}=\left[\cos\left(\frac{\left(j+1\right)\pi}{n}\right)\sigma_{x}+\sin\left(\frac{\left(j+1\right)\pi}{n}\right)\sigma_{z}\right]\otimes\left[\cos\left(\frac{j\pi}{n}\right)\sigma_{x}+\sin\left(\frac{j\pi}{n}\right)\sigma_{z}\right].\label{even-2}
\end{equation}

Further,
\[
X_{n}X_{1}=\widetilde{X_{n}}\otimes\widetilde{X_{1}}
\]

\begin{equation}
=-\cos\left(\frac{\pi}{n}\right)\sigma_{x}\otimes\sigma_{x}-\sin\left(\frac{\pi}{n}\right)\sigma_{x}\otimes\sigma_{z}. \label{even-3}
\end{equation}

Using \ref{even-1}, \ref{even-2} ,\ref{even-3} and basic arithmetics,
the $n$-cycle chained Bell inequality for quantum systems transforms as 

\[
\frac{n}{2}\cos\left(\frac{\pi}{n}\right)\langle\sigma_{x}\otimes\sigma_{x}\rangle+\frac{n}{2}\cos\left(\frac{\pi}{n}\right)\langle\sigma_{z}\otimes\sigma_{z}\rangle+\frac{n}{2}\sin\left(\frac{\pi}{n}\right)\langle\sigma_{x}\otimes\sigma_{z}\rangle-\frac{n}{2}\sin\left(\frac{\pi}{n}\right)\langle\sigma_{z}\otimes\sigma_{x}\rangle\le n-2,
\]
which further simplifies to 
\[
\cos\left(\frac{\pi}{n}\right)\left[\langle\sigma_{x}\otimes\sigma_{x}\rangle+\langle\sigma_{z}\otimes\sigma_{z}\rangle\right]+\sin\left(\frac{\pi}{n}\right)\left[\langle\sigma_{x}\otimes\sigma_{z}\rangle-\langle\sigma_{z}\otimes\sigma_{x}\rangle\right]\le\frac{2\left(n-2\right)}{n}.
\]

For a two qubit density matrix $\rho,$ this translates into 
\begin{equation}
\text{Tr}\left(O_{n}\rho\right)\le\frac{2\left(n-2\right)}{n},\label{Sp1}
\end{equation}
where $O_{n}=\cos\left(\frac{\pi}{n}\right)\left[\sigma_{x}\otimes\sigma_{x}+\sigma_{z}\otimes\sigma_{z}\right]+\sin\left(\frac{\pi}{n}\right)\left[\sigma_{x}\otimes\sigma_{z}-\sigma_{z}\otimes\sigma_{x}\right].$
The condition for violation of n-cycle chained Bell inequality becomes
\begin{equation}
\text{Tr}\left(O_{n}\rho\right)>\frac{2\left(n-2\right)}{n}\label{sp3}
\end{equation}
The eigenvalues of $O_{n}$ are $2,0,0,-2.$ Suppose the eigenvalues
of $\rho$ are $\lambda_{1}\ge\lambda_{2}\ge\lambda_{3}\ge\lambda_{4},$
then
\begin{equation}
\text{Tr}\left(O_{n}\rho\right)\le2\left(\lambda_{1}-\lambda_{4}\right).\label{sp2}
\end{equation}

Using \ref{sp2} and \ref{sp3}, the necessary condition
for the violation of $n$-cycle chained Bell inequality turns out to
be	
$$
\lambda_{1}-\lambda_{4} > \frac{n-2}{n}.
$$
This completes the proof.
\end{proof}
\begin{Remark}
It is easy to see that set of two-qubit quantum states that can violate
chained Bell inequality shrinks as we increase $n$.
In the infinite $n$ scenario, the only two qubit state that violates the
inequality is a Bell state!
\end{Remark}
\begin{figure}[H]
\begin{center}\includegraphics[scale=0.23]{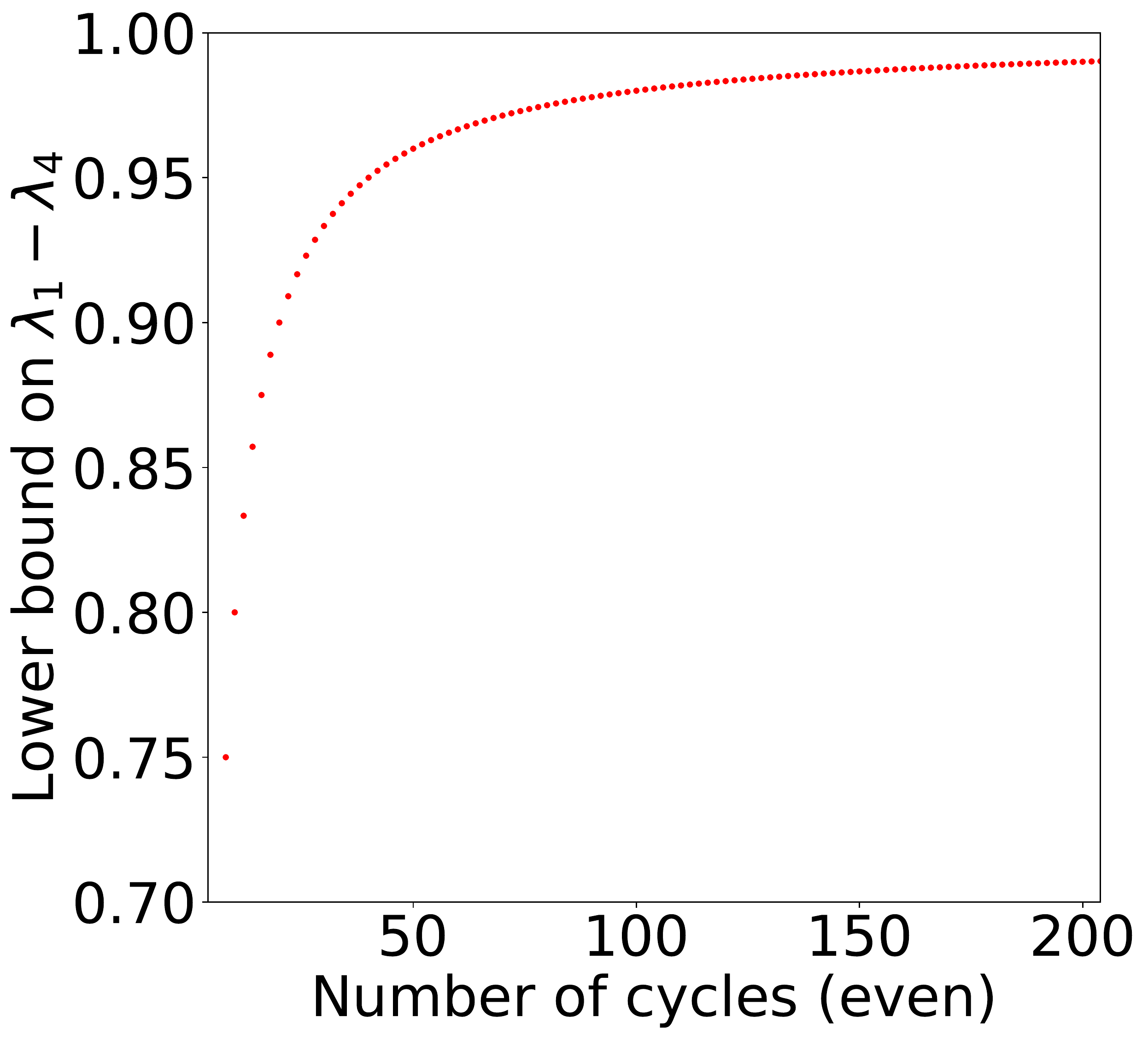}\end{center}\caption{Here we plot the lower bound on the difference of extremal eigenvalues of a two qubit density matrix as a function of even values of $n$. The set of two qubit quantum states, which could potentially violate
chained Bell inequality ( as our  is necessary and not sufficient), shrinks as we increase $n.$
In the infinite $n$ scenario, the only two qubit state that might violate the
inequality is Bell state!}
\end{figure}

Now, we move on to provide a simple extension of even cycle non-contextuality inequalities (as in equation \ref{chained-CHSH}   to the continuous variable case.
%
The even cycle non-contextuality inequalities can be extended to continuous
variables quite easily following the work of Arora {\it{et. al}} \cite{arora2015proposal} where the authors provide the continuous variable extension for $n=4$.
We have already discussed the construction corresponding to the maximal violation of inequality in \eqref{chained-CHSH}.
The inequality is maximally violated by $\left(0,1/\sqrt{2},-1/\sqrt{2},0\right)^T$
and the maximum violation is $n\cos\left(\pi/n\right).$ We define the following non-contextuality
operator in this regard,

\begin{equation}
C_{n}=\sum_{j=1}^{n-1}X_{j}X_{j+1}-X_{n}X_{1}.
\end{equation}

We know that
\begin{equation}
\exp\left(\iota\theta\boldsymbol{n.\sigma}\right)\boldsymbol{\sigma}\exp\left(-\iota\theta\boldsymbol{n.\sigma}\right)=\boldsymbol{\sigma}\cos\left(2\theta\right)+\boldsymbol{n}\times\boldsymbol{\sigma}\sin\left(2\theta\right)+\boldsymbol{n}\text{ }\boldsymbol{n}.\boldsymbol{\sigma}\left(1-\cos\left(2\theta\right)\right)
\end{equation}
For $\sigma=\sigma_{x}\hat{x}$ and $n=\hat{z},$

\begin{equation}
\exp\left(\iota\theta\sigma_{z}\right)\sigma_{x}\exp\left(-\iota\theta\sigma_{z}\right)=\sigma_{x}\cos\left(2\theta\right)+\sigma_{y}\sin\left(2\theta\right)
\end{equation}
Let us look back at the operator in equation \ref{Ops} more closely. This
can be thought of as $\sigma_{x}$ rotated around $z$ axis with angle
$\left(\frac{j\pi}{2n}\right).$ To get the continuous variable representation,
let us start with the quantum mechanical translational operator $\exp\left(\frac{-\iota pL}{\hbar}\right),$
which translates a particle by distance $L$. This operator is not
hermitian and hence we introduce the following symmetric combination to make it hermitian,
\begin{equation}
\mathcal{X}\left(0\right)\equiv\frac{e^{-\iota pL/\hbar}+e^{\iota pL/\hbar}}{2}=\cos\left(\frac{pL}{\hbar}\right).
\end{equation}
Let $U(\phi)=\exp\left(\frac{\iota\mathcal{Z}\phi}{2}\right)$ where
$\mathcal{Z}=\text{sgn}\left(\sin\left(\frac{q\pi}{L}\right)\right).$ One can easily see that

\begin{equation}
\mathcal{X}(\phi)\equiv U^{\dagger}\left(\phi\right)\mathcal{X}(0)U^{\dagger}\left(\phi\right)
\end{equation}
and 
\begin{equation}
\widetilde{X_{j}}=\mathcal{X}\left(\frac{j\pi}{n}\right).
\end{equation}

Let $\phi\left(q\right)=\langle q\vert\phi\rangle$ be the localized
quantum state symmetric about $q=\frac{L}{2},$ for some length scale $L$.
and $\phi_{n}\left(q\right)\equiv\phi\left(q-nL\right).$ Using this
construction, following states are defined:

\begin{equation}
\vert\psi_{0}\rangle\equiv\frac{1}{\sqrt{M}}\sum_{n=-\frac{M}{2}}^{n=\frac{M-1}{2}}\vert\phi_{2n+1}\rangle
\end{equation}


\begin{equation}
\vert\psi_{1}\rangle\equiv\frac{1}{\sqrt{M}}\sum_{n=-\frac{M}{2}}^{n=\frac{M-1}{2}}\vert\phi_{2n}\rangle
\end{equation}
Let $ \vert\psi_{+}\rangle\equiv\frac{\vert\psi_{0}\rangle+\vert\psi_{1}\rangle}{\sqrt{2}}$
and  $\vert\psi_{-}\rangle\equiv\frac{\vert\psi_{0}\rangle-\vert\psi_{1}\rangle}{\sqrt{2}}.$
 Interestingly, for $N=2M$,
\begin{equation}
\langle\psi_{+}\vert\mathcal{X}\vert\psi_{+}\rangle=\left(\frac{N-1}{N}\right),
\end{equation}
and
\begin{equation}
\langle\psi_{-}\vert\mathcal{X}\vert\psi_{-}\rangle=-\left(\frac{N-1}{N}\right).
\end{equation}
The appropriate entangled state which shows the violation is
\begin{equation}
\vert\psi\rangle\equiv\frac{\vert\psi_{+}\rangle_{1}\vert\psi_{-}\rangle_{2}-\vert\psi_{-}\rangle_{1}\vert\psi_{+}\rangle_{2}}{\sqrt{2}}\label{State}.
\end{equation}
Let us calculate the violation for the state in \ref{State}.

\begin{equation}
\langle\mathcal{X}(\phi)\otimes\mathcal{X}(\theta)\rangle=-\left(\frac{N-1}{N}\right)^{2}\cos\left(\phi-\theta\right)
\end{equation}

\begin{equation}
\langle C_{n}\rangle=-\left(\frac{N-1}{N}\right)^{2}\left[(n-1)\cos\left(\pi/n\right)-\left\{ -\cos\left(\pi/n\right)\right\} \right]
\end{equation}

\begin{equation}
=-\left(\frac{N-1}{N}\right)^{2}n\cos\left(\pi/n\right)
\end{equation}
For large $N,$
\begin{equation}
\frac{N-1}{N}\rightarrow1,
\end{equation}
and hence we get the maximum quantum violation $n\cos\left(\frac{\pi}{n}\right).$ 
The experimental implementation of the continuous variable extension is quite simple and follows directly from the work of Arora {\it{et. al.}} \cite{arora2015proposal}.
\section{Simulating PR Box}\label{PR}
The KS box is a powerful resource which can be used to efficiently simulate the most non-local no-signalling box, i.e. PR box \cite{popescu1994quantum}. The PR box has initially been defined as the box which allows maximum violation of the CHSH inequality in no-signalling theories. One can generalise the notion of PR box corresponding to chained Bell inequalities.
\begin{Definition}
A PR-box is a no-signalling resource with input pair $x, y$ and corresponding output pair $a, b$ where each of these variables takes their values from the set $\left\{ 0,1\right\}.$ The statistics of the PR box follows the following relation:
\begin{equation}
xy=a\oplus b,
\end{equation}
which means that the outputs are different if and only if the inputs are $x=y=1,$ otherwise the outputs are same.  The PR box can be generalised for input pair $(x,y)\in \{1,2,\cdots, n\}^2$ and  output from the set  $\left\{ 0,1\right\}$ such that outputs are same when inputs are anything except $\{1,1\}.$ When inputs are $\{1,1\},$ the outputs must be different.
\end{Definition}

Now suppose Alice and Bob are equipped with an arbitrary dimensional KS box. The following table gives the joint probabilities for an $n$-dimensional $KS_p$ box.
\begin{table}[h!]
\begin{center}
\begin{tabular}{|ll||ll|ll|ll|lll|} \hline
   &$x$&$1$& &$2$ & &$\hdots$ & &$n$&&\\
   y&&&&&&&&&&\\
  \hline\hline $1$&&$1-p$ &$0$ &$1-2p$ &$p$ & $\ddots$ &&$1-2p$&$p$& \\
  &&$0$ &$p$ &$p$ &$0$ & & &$p$&$0$&\\\hline

 $2$&&$1-2p$ &$p$ &$1-p$ &$0$ & $\ddots$&  &$1-2p$&$p$&\\
  &&$p$ &$0$ &$0$ &$p$& &   &$p$&$0$&\\\hline

$\vdots$& & $\ddots$ &&  $\ddots$&&  $\ddots$ &&  $\ddots$&&\\\hline

  $n$&&$1-2p$ &$p$ &$1-2p$ &$p$ & $\ddots$&  &$1-p$&$0$&\\
  &&$p$ &$0$ &$p$ &$0$& &   &$0$&$p$&\\\hline

\hline
\end{tabular}
\end{center}
 \caption{The table displays the joint probabilities for an n-dimensional KS$_{p}$-box. Note that each of the blocks along the diagonal are same and similarly all the off diagonal blocks are same. Within a block, the top left element is the probability of getting $(0,0)$, top right signifies the probability of getting $(0,1)$, bottom left indicates the corresponding value for $(1,0)$ and, the probability for $(1,1)$ is indicated by the bottom right entry.  }
\end{table}

KS Box is more powerful than PR box and can be used to simulate the same \cite{bub2009contextuality}.
We ask whether Alice and Bob can simulate a generalised PR box (as defined before) using $KS_{p}$ box. The answer is in the affirmative, and we provide a simple strategy to do so.
\begin{Proposition}
A PR box of dimension (number of inputs for each party) $n$ can be simulated efficiently using a KS
box of dimension $2n-1$ with marginal value of $p=\frac{1}{2}$.
\end{Proposition}
\begin{proof}[Proof]

To prove our claim, we provide the following strategy: Alice relabels her inputs for PR box as follows:
$$1 \rightarrow 1, 2 \rightarrow 2, 3 \rightarrow 4, 4 \rightarrow 6 \cdots,  n \rightarrow 2n-2.$$
Similarly, Bob relabels his inputs as follows:
$$1 \rightarrow 1, 2 \rightarrow 3, 3 \rightarrow 5, 4 \rightarrow 7 \cdots,  n \rightarrow 2n-1.$$
The relabelled inputs are used as fresh input for the $KS_{\frac{1}{2}}$ box. Alice outputs what she gets as output from the $KS_{\frac{1}{{2}}}$ . Bob flips his output from $KS_{\frac{1}{2}}$ box in every round and outputs the resultant value. This strategy simulates the statistics corresponding to generalised PR box.
\end{proof}

%
%
%
%
Given the even cycle generalisation of CHSH inequality, the marginal probabilities $p$ in the $KS_{p}$ required to saturate classical bound, quantum bound  and no-signalling bound are given by 
$$
p_{\text{c}}\le\frac{n-2}{2(n-1)}, 
$$
$$
p_{\text{q}}\le\frac{n\left(\text{\ensuremath{\cos}}\left(\frac{\pi}{n}\right)+1\right)-2}{4(n-1)}
$$
and 
\begin{equation}
p_{\text{NS}}\le\frac{1}{2}
\end{equation}
respectively.
\begin{Remark}
For a large value of $n$, all the above probability expressions tend
to one half. However, the quantum probability approaches the PR box
limit of $\frac{1}{2}$ significantly faster than the classical
probability.For large n, all these probabilities approach $\frac{1}{2}.$
\end{Remark}

\begin{figure}[H]
\begin{center}\includegraphics[scale=0.23]{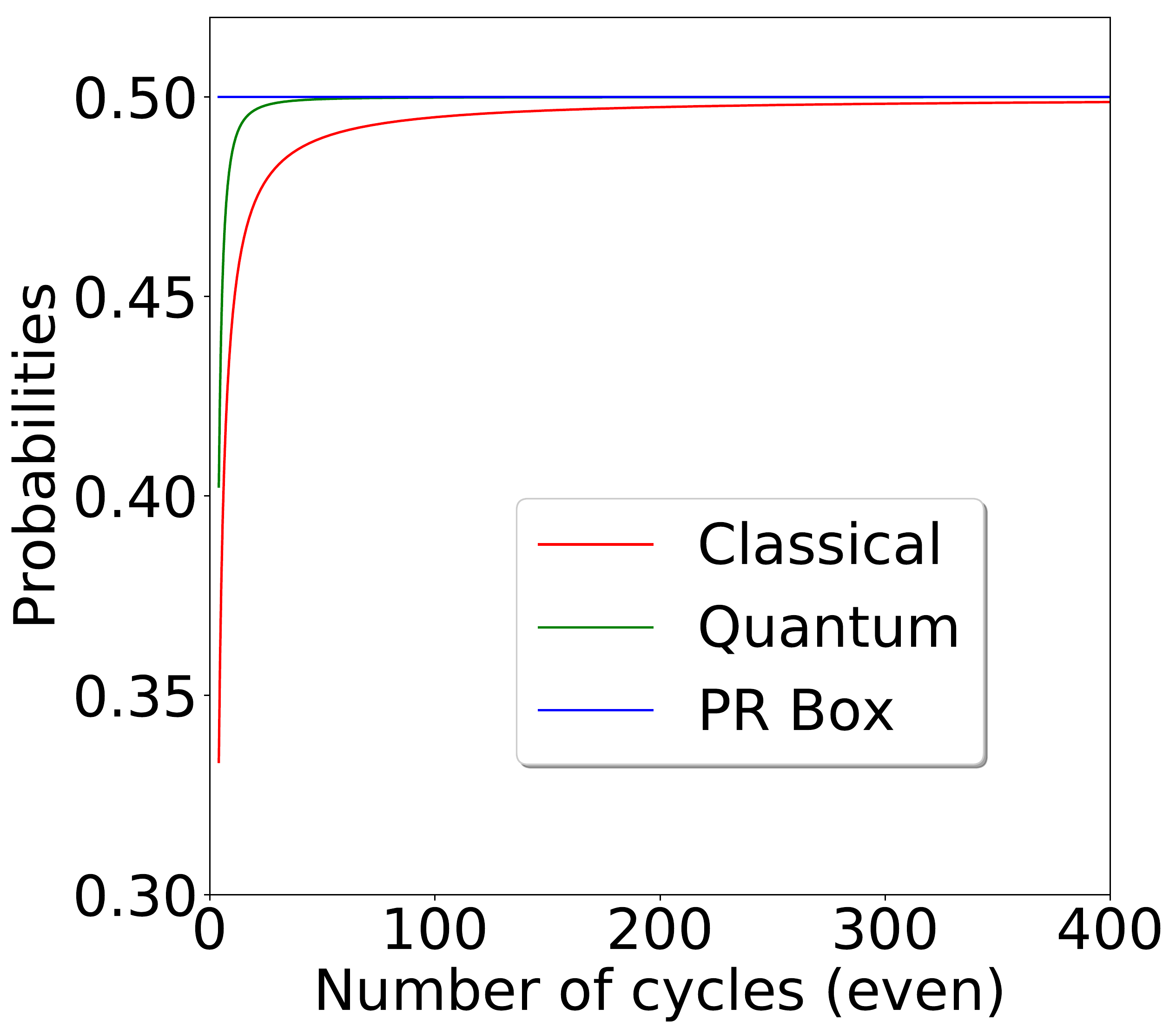}\end{center}\caption{We look at KS box probabilities in various regimes. Note that the quantum probability approaches the PR box limit faster than classical probability as we increase the number of cycles.}
\end{figure}

\section{Conclusion} \label{conclusion}
To conclude, we studied arbitrary dimensional KS box, generalised PR box and $n$-cycle non-contextuality inequalities in this work. We provided the optimal classical strategy and the corresponding success probability for classically simulating the KS box. For future work, it is worthwhile exploring the optimal quantum strategy for this purpose. We provided the sufficient condition for the violation of the generalised KCBS inequality and necessary condition for the violation of even-cycle generalisation of CHSH inequality. We also discussed the continuous variable extension of even-cycle generalisation of CHSH inequality. We leave the continuous variable extension of KCBS and generalised KCBS inequality for future work. 
We also studied the strategy for simulating a generalised PR box using KS box.  It is also interesting to explore further how the generalised PR box, arbitrary dimensional KS box and $n$-cycle non-contextuality inequalities are related to each other and their implications.

\acknowledgments{We thank Atul Singh Arora, Naresh Boddu and Debasis Mondal for useful discussions. KB and MR acknowledge the CQT Graduate Scholarship. KB, MR and LCK are thankful to the Natural Research Foundation and the Ministry of Education, Singapore for financial support.}

\reftitle{References}
%
%
\externalbibliography{yes}
\bibliography{KS}
%
%
%
\end{document}